\documentclass[aps,pra,twocolumn,superscriptaddress,amsmath,amssymb]{revtex4-2}

\usepackage{graphicx}
\usepackage{bm}
\usepackage{physics}
\usepackage{hyperref}
\usepackage{mathtools}
\usepackage{braket}
\usepackage{times}
\usepackage{quantikz}
\usepackage{amsthm}
\newtheorem{lemma}{Lemma}
\newtheorem{proposition}{Proposition}

\newcommand{\Cr}{C_{\mathrm r}}
\newcommand{\diag}{\mathrm{diag}}

\begin{document}

\title{Coherence Scaling  in Quantum Communication Protocols}

\author{Pedro Henrique Alvarez}
\affiliation{Instituto de F\'{\i}sica Gleb Wataghin, Universidade Estadual de Campinas (UNICAMP), Campinas, SP, Brazil}
\affiliation{Institut für Physik, Carl-von-Ossietzky Universität Oldenburg, Oldenburg, Germany}

\author{Marcos C\'esar de Oliveira}
\email{marcos@ifi.unicamp.br}
\affiliation{Instituto de F\'{\i}sica Gleb Wataghin, Universidade Estadual de Campinas (UNICAMP), Campinas, SP, Brazil}

\date{\today}

\begin{abstract}
We investigate how quantum coherence scales and is redistributed in quantum
communication protocols, using superdense coding and quantum teleportation as
paradigmatic case studies. Employing the relative entropy of coherence as a
circuit-level resource measure, we show that multipartite resource states
relevant to generalized superdense coding can enable scalable communication
while exhibiting only logarithmic or even constant coherence growth, depending
on their entanglement structure. In sharp contrast, quantum teleportation
displays an unavoidable, protocol-induced coherence cost that grows linearly
with the number of teleported qubits and is independent of the input state.
Through a stage-resolved analysis of the teleportation circuit, we separate
protocol-generated coherence from message-dependent contributions and identify
a universal two-bit coherence offset per teleported qubit at the
maximal-coherence stage. We further demonstrate explicitly that this extensive
intermediate coherence generation is fully consistent with
information-theoretic bounds, including the Holevo limit, and does not
correspond to an increase in accessible classical information.
\end{abstract}

\maketitle

\section{Introduction}
Quantum communication protocols exploit nonclassical resources to transmit
information in ways that are impossible or inefficient classically.
Superdense coding enables the transmission of two classical bits by sending a
single qubit, assisted by shared entanglement, while quantum teleportation
allows the faithful transfer of an unknown quantum state using shared
entanglement and classical communication
\cite{BennettWiesner1992,Bennett1993,NielsenChuang2011,SchumacherWestmoreland2010Book}.
These protocols are cornerstones of quantum information theory and serve as
paradigmatic examples of how quantum resources can be harnessed for
communication tasks.

Entanglement is usually regarded as the central enabling resource in these
protocols. However, entanglement alone does not fully characterize the
operational structure of their circuit implementations. In particular, the
execution of quantum communication protocols involves the generation,
redistribution, and destruction of superpositions at intermediate stages, even
when the net effect of the protocol corresponds to a simple communication
primitive, such as an identity channel. Capturing these dynamical aspects
requires a resource concept that is sensitive to basis-dependent superposition
structure and that can be tracked throughout the protocol execution.

Quantum coherence, defined as the presence of superpositions with respect to a
preferred basis, has emerged as an independent resource within a well-defined
resource-theoretic framework \cite{Baumgratz2014}. Although coherence is
inherently basis-dependent, in gate-based quantum communication protocols the
computational basis is operationally fixed by the encoding, gate, and
measurement structure. In this setting, coherence acquires a concrete
operational meaning: it quantifies the transient superposition structure
explored by a circuit as it realizes a given communication primitive.

Previous studies have shown that coherence depletion is a common feature of
quantum algorithms and can reveal structural properties that are not captured
by entanglement measures alone
\cite{Shi2017GroverCoherence,Liu2019EntropyCoherenceDepletion}. These results
suggest that coherence can act as a diagnostic of internal protocol structure.
Despite this progress, a systematic understanding of how coherence scales in
quantum communication protocols, how it is redistributed during protocol
execution, and how large intermediate coherence can be reconciled with
information-theoretic limits remains underdeveloped.

Recent work has extended the resource theory of quantum coherence beyond static state manipulation, emphasizing dynamic and process-level descriptions relevant to communication protocols. In particular, modern resource-theoretic analyses have clarified tradeoffs between coherence and classical distinguishability \cite{PhysRevA.111.062215} and introduced frameworks for one-shot and process-level coherence manipulation \cite{Luo2025DynamicCoherence,Takagi2024VirtualResource}. These developments complement growing interest in distributed quantum information processing and quantum networks, where teleportation-based primitives and feed-forward control play a central role ( (see, e.g. \cite{DistributedQIP2025Review}.

In this work we address this gap by developing a coherence-based analysis of
quantum communication protocols at the circuit level. We focus on how coherence
scales with system size, how it is redistributed across subsystems during
protocol execution, and how protocol-induced coherence can be separated from
message-dependent contributions. While our discussion is anchored in two
paradigmatic examples—superdense coding and quantum teleportation—we also
formulate general coherence accounting principles applicable to a broad class
of LOCC-assisted quantum communication protocols.

Our main contributions are fourfold. First, we analyze coherence scaling for
multipartite resource states relevant to generalized superdense coding and show
that such states can support scalable communication tasks while exhibiting only
logarithmic or even constant coherence growth. Second, we develop a
stage-resolved analysis of quantum teleportation that reveals an unavoidable
and extensive coherence cost intrinsic to the protocol, independent of the
input state. Third, we derive general bounds showing that large intermediate
coherence can arise purely from protocol-induced branching and conditional
operations, even when the implemented channel is information-neutral. Fourth,
we demonstrate explicitly that this extensive coherence generation is fully
consistent with information-theoretic bounds, including the Holevo limit,
thereby clarifying the conceptual distinction between coherence as a
circuit-level resource and accessible classical information.

Together, these results position coherence as a powerful tool for analyzing the
internal resource structure and superposition complexity of quantum
communication protocols. Beyond the specific protocols studied here, this
perspective provides a framework for understanding coherence overheads in more
general communication settings, including networked and modular architectures
relevant to emerging quantum internet technologies.

This paper is organized as follows. In Sec.~II we introduce the relative entropy
of coherence as a circuit-level resource and summarize the properties relevant
for scalable communication protocols. In Sec.~III we analyze coherence scaling
in superdense coding, focusing on multipartite resource states and contrasting
logarithmic and constant coherence behavior. Section~IV presents a stage-resolved
analysis of quantum teleportation, including coherence redistribution,
parallel-teleportation scaling laws, and the effects of imperfect resources.
In Sec.~V we reconcile extensive protocol-induced coherence with
information-theoretic bounds and clarify its relation to the Holevo limit.
Finally, Sec.~VI summarizes our results and discusses their broader implications
for general quantum communication protocols and networked architectures.

\section{Coherence as a Resource}
\label{sec:coherence}

We quantify quantum coherence using the relative entropy of coherence \cite{Baumgratz2014, PhysRevA.111.062215},
\begin{equation}
\Cr(\rho) = S(\rho_{\diag}) - S(\rho),
\label{eq:Cr}
\end{equation}
where $S(\rho)=-\Tr(\rho\log_2\rho)$ is the von Neumann entropy and $\rho_{\diag}$ is obtained by deleting
all off-diagonal elements of $\rho$ in the computational basis.

Among the various coherence measures proposed in the literature, the relative entropy of coherence is
particularly well suited to the present analysis for two reasons. First, it admits a direct
information-theoretic interpretation as the distinguishability between $\rho$ and its fully dephased
counterpart. Second, it is additive under tensor products, a property essential for analyzing scalable,
multi-register communication protocols and for isolating protocol-induced versus message-dependent
contributions.

Although coherence is inherently basis dependent, this dependence is not a drawback here: the
computational basis is operationally fixed by the encoding, gate, and measurement structure of the
protocol. Coherence with respect to this basis therefore captures physically meaningful superpositions
that are generated, redistributed, and destroyed during protocol execution. From this perspective,
coherence complements entanglement rather than competing with it. While entanglement quantifies
nonclassical correlations between subsystems at a given stage, coherence tracks the superposition
structure of the global circuit state and provides information about the dynamical resources required
to implement a protocol, even when the net action is simple.

\begin{lemma}[Additivity of the relative entropy of coherence]
For any two quantum states $\rho$ and $\sigma$,
\begin{equation}
\Cr(\rho\otimes\sigma)=\Cr(\rho)+\Cr(\sigma).
\end{equation}
\end{lemma}

This additivity property will play a central role in the analysis of parallel teleportation, where it
enables a clean separation between protocol-induced coherence costs and message-dependent contributions.

\section{Superdense Coding and Coherence Scaling}
\label{sec:sd}
In superdense coding, Alice and Bob share an entangled resource state. By applying a set of local
unitaries to her subsystem, Alice encodes classical messages, which are subsequently decoded by Bob via
a joint measurement after transmission of a subset of the qubits \cite{BennettWiesner1992,NielsenChuang2011}.
In generalized multipartite versions, the number of classical bits transmitted grows linearly with the
number of qubits sent, provided sufficiently entangled resource states are available.

Multipartite entanglement is not a monolithic resource. It exhibits inequivalent classes that cannot be
interconverted by stochastic local operations and classical communication, most notably the GHZ and $W$
families \cite{DurVidalCirac2000}. While these classes are often discussed in terms of entanglement, they
also possess markedly different coherence structures with respect to the computational basis, which is
operationally fixed by the encoding and decoding operations in superdense coding.

For the $n$-qubit $W$ state,
\begin{equation}
\ket{W_n}=\frac{1}{\sqrt{n}}\sum_{k=1}^n \ket{0\cdots1_k\cdots0},
\end{equation}
the relative entropy of coherence is exactly
\begin{equation}
\Cr(\ket{W_n}\bra{W_n})=\log_2 n,
\end{equation}
reflecting the fact that the state is supported on $n$ computational-basis components with equal amplitudes.
By contrast, the $n$-qubit GHZ state,
\begin{equation}
\ket{GHZ_n}=\frac{1}{\sqrt{2}}\left(\ket{0}^{\otimes n}+\ket{1}^{\otimes n}\right),
\end{equation}
has constant coherence $\Cr(\ket{GHZ_n}\bra{GHZ_n})=1$, independent of system size.

\begin{figure}[t]
\centering
\includegraphics[width=0.9\columnwidth]{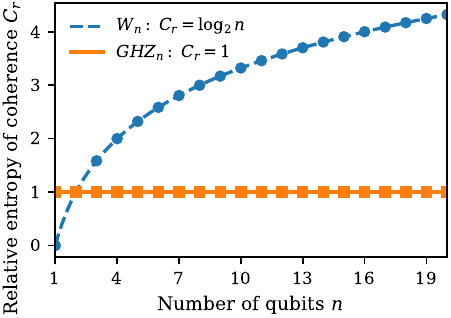}
\caption{Relative entropy of coherence for two paradigmatic multipartite resource
states relevant to generalized superdense coding. The $W_n$ state exhibits
logarithmic scaling, $C_{r}=\log_2 n$, while the $GHZ_n$ state has constant
coherence $C_{r}=1$ independent of $n$. This illustrates that coherence scaling
need not track communication capability.}
\label{fig:Wn}
\end{figure}

Fig.~\ref{fig:Wn}  illustrates a key conceptual point: multipartite resource states capable of
supporting scalable communication tasks can exhibit drastically different coherence scaling. In particular, the
amount of coherence present in the shared resource state does not track either the size of the system or the
communication capacity in a straightforward way. Superdense coding therefore provides a natural counterpoint to
teleportation: whereas teleportation necessarily generates extensive intermediate coherence as part of its circuit
implementation, superdense coding can rely on resource states whose coherence grows only logarithmically or remains
constant. This contrast underscores that coherence is a protocol-dependent quantity whose operational role must be
assessed relative to the specific circuit realization of the communication task.

\section{Quantum Teleportation: Stage-Resolved Coherence}
\label{sec:tp}
Quantum teleportation enables the faithful transfer of an unknown quantum state using shared entanglement and
classical communication, without physically transmitting the quantum system itself \cite{Bennett1993,NielsenChuang2011}.
Although teleportation is often described as implementing an effective identity channel on the input state, the
circuit realization of this channel involves nontrivial intermediate dynamics, including entangling operations and
measurements.

To uncover the coherence resources required by teleportation, we analyze the protocol in a stage-resolved manner.
We introduce a stage index defined by fixed circuit cuts after the application of each gate, and evaluate the relative entropy of coherence of the global
circuit state at each stage, as represented in Fig. \ref{teleport1}.  Measurement is treated non-selectively, i.e., as complete dephasing in the measurement
basis, which captures operational coherence loss without conditioning on outcomes. This approach allows us to
distinguish coherence generated and consumed by the protocol itself from coherence carried by the input message.

\begin{figure}\centering
\includegraphics[width=8.0cm]{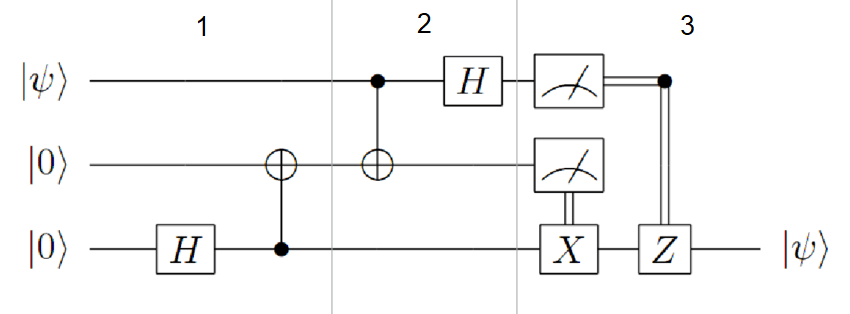}
 \caption{Quantum circuit representing the protocol for quantum teleportation based on \cite{NielsenChuang2011}, the gray vertical lines separate the three stages of the process, first creating the EPR pair, second the teleportation process and the third and last phase the message extraction.}
\label{teleport1}\end{figure}
\subsection{Total coherence across stages}

\begin{figure}[t]
\centering
\includegraphics[width=\columnwidth]{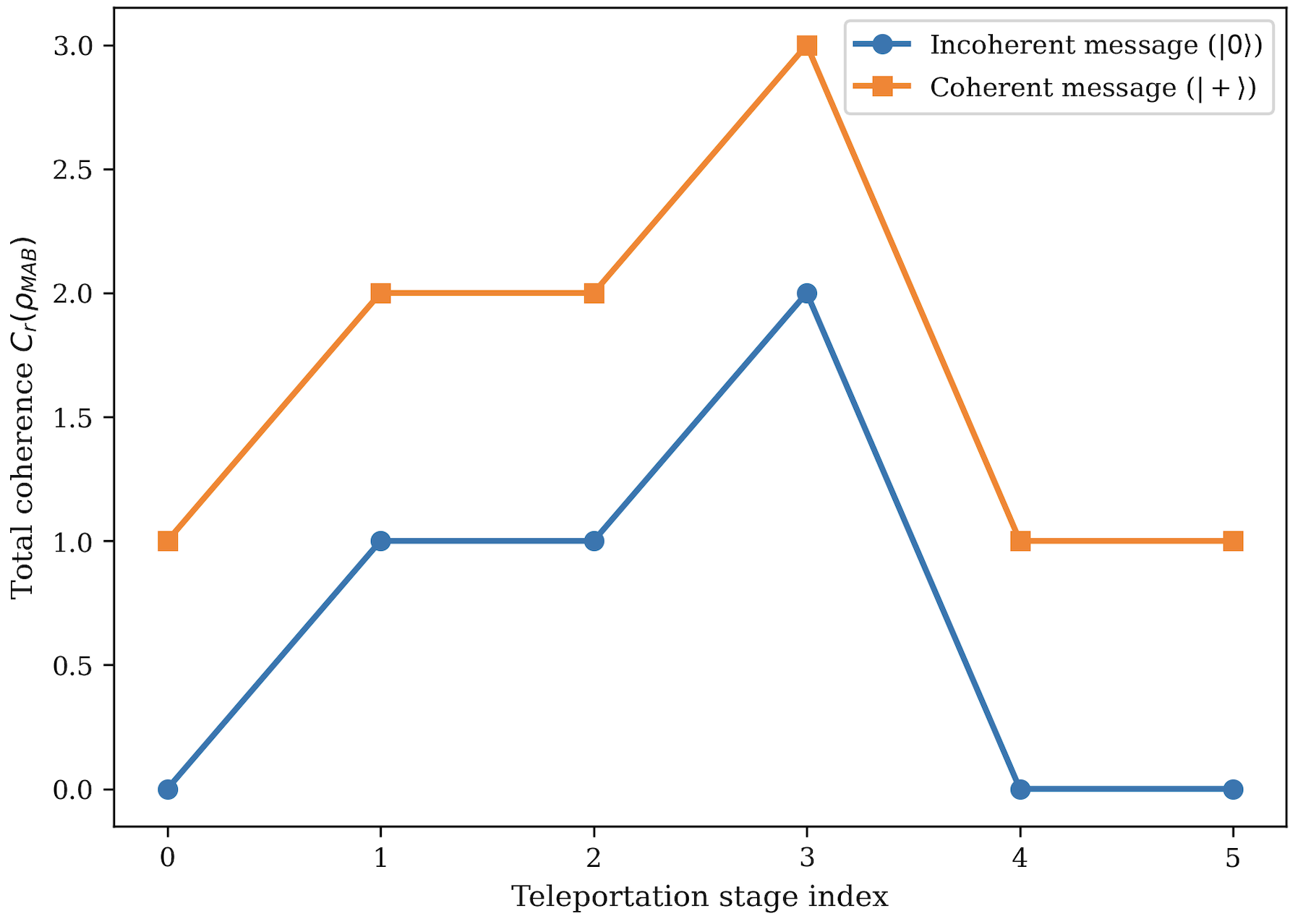}
\caption{Total relative entropy of coherence across teleportation stages. Coherence is generated by entangling gates,
peaks immediately before measurement, and is suppressed by non-selective measurement.}
\label{fig:tp_total}
\end{figure}

Figure~\ref{fig:tp_total} reveals a characteristic coherence profile intrinsic to teleportation. Even when the input
state is incoherent, the total coherence increases substantially due to entangling operations that correlate the message
system with the shared Bell pair. This buildup culminates at a maximal stage immediately prior to measurement, after
which coherence is sharply reduced by the dephasing induced by non-selective measurement.

Importantly, this behavior reflects a structural property of teleportation circuits: implementing an identity channel
on an unknown quantum state requires the generation of global superpositions at intermediate stages. In this sense,
teleportation exhibits an unavoidable, protocol-induced coherence cost that is independent of the message being
transmitted.

\subsection{Redistribution and partial coherences}

\begin{figure}[t]
\centering
\includegraphics[width=\columnwidth]{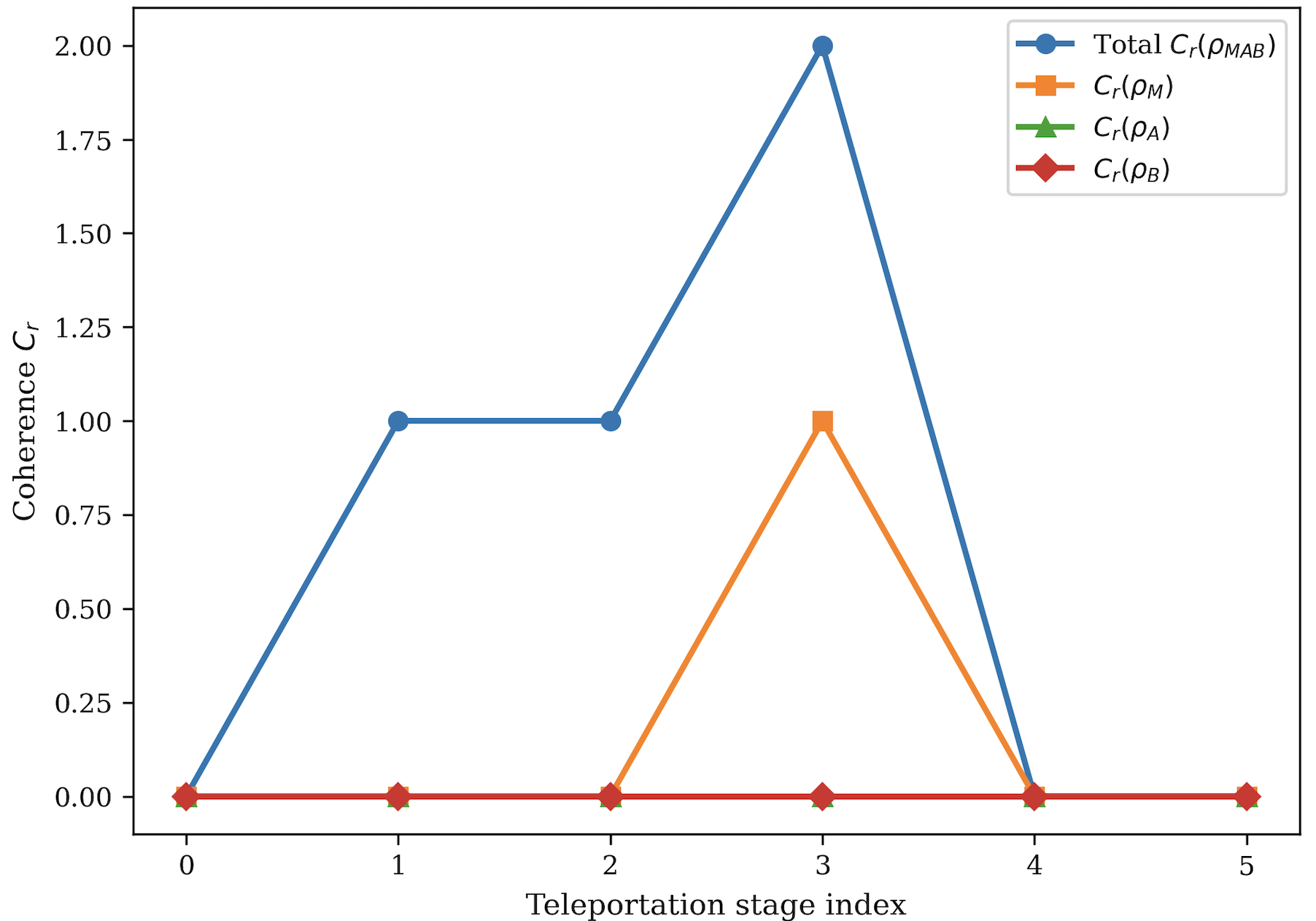}
\caption{Single-qubit reduced-state coherences during teleportation of an incoherent input state $\ket{0}$, illustrating
redistribution of coherence among subsystems.}
\label{fig:tp_partials}
\end{figure}

While Fig.~\ref{fig:tp_total} captures total coherence of the global circuit state, it does not indicate where this
coherence resides. To address this question, we examine reduced-state coherences of individual subsystems during the
protocol. Figure~\ref{fig:tp_partials} shows that, despite the absence of initial coherence in the input message,
coherence becomes redistributed among the message qubit and the entangled ancillae.

A key observation is that the sum of local coherences does not account for the total coherence. Most coherence generated
during teleportation remains delocalized and encoded in genuinely multipartite superpositions. This delocalization is
not merely a bookkeeping subtlety: it is tightly connected to the conditional structure of teleportation. The circuit
creates correlations that allow a classical measurement outcome to select the appropriate correction on Bob’s side, and
the global superposition structure is the mechanism by which this conditional reconstruction becomes possible. In this
sense, delocalized coherence is an internal resource enabling the protocol’s classical control logic.

\subsection{Scaling law for parallel teleportation}

We now turn to the scaling behavior of coherence when teleporting multi-qubit messages. Consider teleporting an $n$-qubit
state using $n$ independent teleportation gadgets (complete circuit) in parallel. At any fixed stage of the circuit, the global unitary
evolution factorizes into $n$ identical components acting on disjoint registers, together with the message state.

\begin{proposition}[Stage-resolved coherence decomposition]
At any fixed stage $i$,
\begin{equation}
\Cr^{TP_i}(\ket{\psi_n}\bra{\psi_n}) = n\,\Cr^{TP_i}(\ket{0}\bra{0}) + \Cr(\ket{\psi_n}\bra{\psi_n}).
\end{equation}
\end{proposition}
This decomposition isolates two distinct contributions -- a protocol-induced term that scales extensively with the number
of teleportation gadgets ($C_r^{TP_i}$), and a message-dependent term that depends only on the intrinsic coherence of the input state, ($C_r$).

At the maximal-coherence stage, the single-gadget contribution takes the universal value $\Cr^{TP}(\ket{0}\bra{0})=2$, as we saw in Fig. \ref{fig:tp_partials}, yielding
\begin{equation}
\Cr^{TP}(\ket{\psi_n}\bra{\psi_n}) = 2n + \Cr(\ket{\psi_n}\bra{\psi_n}).
\label{eq:scaling}
\end{equation}
This result shows that teleportation incurs a fixed coherence overhead of two bits per teleported qubit, regardless of
the structure of the input state.

\paragraph*{Effect of decoherence and imperfect resource states.}
The scaling law in Eq.~\eqref{eq:scaling} is derived for ideal Bell pairs and noiseless gates. In the presence of
decoherence or imperfect entangled resources, the total coherence generated during teleportation is reduced
quantitatively but not qualitatively altered. Imperfect Bell pairs decrease the maximal attainable coherence per gadget,
effectively renormalizing the constant offset in Eq.~\eqref{eq:scaling}, while preserving linear scaling with the number
of teleported qubits as long as the protocol remains operational. Local decoherence and dephasing primarily suppress
delocalized coherence at intermediate stages, leading to a reduced peak height in Figs.~\ref{fig:tp_total} and \ref{fig:tp_partials}, but do not
eliminate the protocol-induced character of the coherence cost. This highlights that extensive coherence generation is a
structural feature of teleportation circuits rather than an artifact of idealized resources.

\begin{figure}[t]
\centering
\includegraphics[width=0.9\columnwidth]{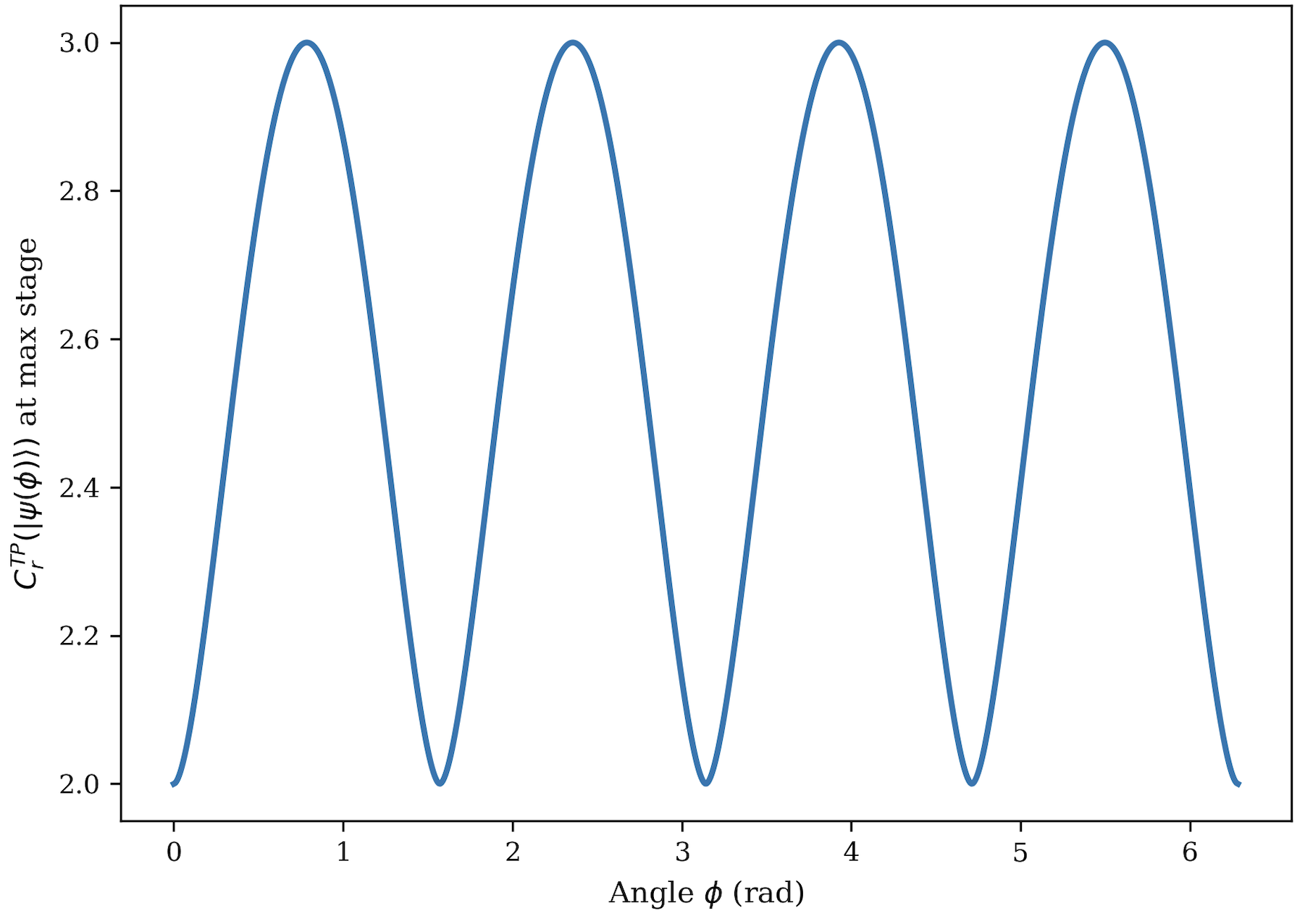}
\caption{Message-dependent coherence contribution for teleportation of a single-qubit state
$\ket{\psi(\phi)}=\cos\phi\ket{0}+\sin\phi\ket{1}$.}
\label{fig:tp_angle}
\end{figure}

As a simple illustration in Figure~\ref{fig:tp_angle} we plot the additive nature of the message-dependent term, from Eq. (\ref{eq:scaling}) for an arbitrary qubit state \begin{equation}
   \ket{\psi} = \cos{\frac{\theta}{2}}\ket{0} + e^{i\gamma}\sin{\frac{\theta}{2}}\ket{1},
\end{equation} 
taking $\gamma=0$, without loss of generality, and defining $\phi=\theta/2$. The variation with $\phi$
reflects solely the coherence of the input state, while the constant offset corresponds to the protocol-induced
contribution. This separation underscores the universality of the coherence cost associated with teleportation and
clarifies that extensive coherence generation is a property of the protocol itself, not of the information being
transmitted.

\section{Consistency with Information-Theoretic Bounds}
\label{sec:holevo}
The extensive coherence generated during teleportation might at first sight appear to conflict with information-theoretic
limits on communication. In particular, the linear scaling $\Cr^{TP}(\ket{\psi_n}\bra{\psi_n}) = 2n + \Cr(\ket{\psi_n}\bra{\psi_n})$ could be
misread as an unbounded growth of informational content. Here we show explicitly that no such violation occurs and
clarify the operational meaning of coherence in relation to the Holevo bound.

For an ensemble of quantum states $\mathcal{E}=\{p_x,\rho_x\}$ encoding a classical random variable $X$, the accessible
classical information, $I_{\mathrm{acc}}$,   obtainable by any measurement is bounded by the Holevo quantity
\cite{Holevo1973,NielsenChuang2011}
\begin{equation}
I_{\mathrm{acc}} \le \chi(\mathcal{E})
= S\!\left(\sum_x p_x\rho_x\right)-\sum_x p_x S(\rho_x).
\label{eq:holevo}
\end{equation}
Crucially, $\chi$ depends only on the ensemble of states available at the output and is insensitive to intermediate
circuit dynamics.

Teleportation implements an identity channel on quantum states. If Alice prepares an ensemble $\mathcal{E}$ prior to
teleportation, the ensemble received by Bob is identical, $\rho_x \mapsto \rho_x$, and therefore
\begin{equation}
\chi_{\mathrm{out}} = \chi_{\mathrm{in}}.
\end{equation}
This equality holds irrespective of the amount of coherence generated at intermediate stages.

The origin of the apparent tension is that the relative entropy of coherence quantifies basis-dependent superposition
structure of a quantum state rather than ensemble distinguishability. During teleportation, large values of $\Cr$ arise
from global superpositions involving the message system and the entangled ancillae. These superpositions are the
mechanism by which measurement outcomes become correlated with the appropriate correction on Bob’s side, but they are
rendered operationally inaccessible by the measurement-induced dephasing inherent to the protocol.

This distinction is particularly clear when contrasting teleportation with superdense coding. In superdense coding, Alice
actively maps classical messages onto a set of orthogonal quantum states, thereby increasing ensemble distinguishability
and saturating the Holevo bound. In teleportation, by contrast, no new distinguishability is created: the protocol
faithfully transfers the original quantum encoding without amplifying its classical information content.

From a broader viewpoint, the extensive coherence generated during teleportation quantifies the transient superposition
complexity explored by the circuit in order to realize an identity channel on an unknown quantum input. It therefore
captures an internal dynamical resource requirement rather than any increase in information transmissibility, and is
fully consistent with the Holevo bound.

\begin{lemma}[Coherence is not an information monotone]
The relative entropy of coherence is not monotonic under completely positive trace-preserving maps that involve classical
communication and conditional operations. In particular, protocols such as teleportation can generate large intermediate
coherence without increasing the Holevo quantity of any input ensemble.
\end{lemma}

This lemma formalizes the conceptual distinction between coherence as a circuit-level resource and information capacity
as constrained by information-theoretic bounds.

\section{General coherence accounting for arbitrary communication protocols}
\label{sec:general}
Given the simple fundamental facts, now we show that the stage-resolved approach extends beyond teleportation and superdense coding.
In this section we formulate a general coherence accounting framework for an
arbitrary LOCC-assisted quantum communication protocol and derive bounds that
separate (i) message-dependent coherence from (ii) protocol-induced coherence
associated with branching, classical feed-forward, and multipartite control.

\subsection{Protocol model and stage states}
\label{subsec:protocol_model}
Consider a generic communication protocol $\Pi$ acting on a message register $M$
(initial state $\rho_M$), a set of resource registers $R$ shared between nodes
(initial state $\sigma_R$), and (optionally) ancillary work registers $A$
initialized incoherently in the computational basis. The protocol consists of
an alternating sequence of (i) global or local unitaries $U_i$ (possibly
conditioned on a classical record) and (ii) non-selective measurements, which we
model as complete dephasing channels $\Delta$ in the computational basis of the
measured registers. The stage-$i$ circuit state is
\begin{equation}
\rho_i = \Lambda_i\!\left(\rho_M \otimes \sigma_R \otimes \ket{0}\!\bra{0}_A\right),
\label{eq:rho_i_general}
\end{equation}
where $\Lambda_i$ is the (generally LOCC) completely positive trace-preserving
(CPTP) map representing the first $i$ steps of the protocol.

We define the \emph{coherence profile} of the protocol as the sequence
\begin{equation}
{\Cr}_{i} := \Cr(\rho_i), \qquad i=0,1,\dots,
\end{equation}
and the \emph{peak coherence cost} (a protocol-complexity functional) as
\begin{equation}
\mathcal{C}_{\max}(\Pi;\rho_M,\sigma_R) := \max_i \Cr(\rho_i).
\label{eq:Cmax_def}
\end{equation}
In the teleportation analysis, $\mathcal{C}_{\max}$ is attained at the stage
immediately prior to measurement.

\subsection{A branching-entropy identity for coherent control}
\label{subsec:branching_identity}
A key mechanism behind large intermediate coherence in LOCC protocols is
\emph{branching} -- an (implicit) coherent superposition over classical outcomes
that later becomes a classical record via dephasing (measurement). The following identity
makes this precise.

\begin{lemma}[Coherence of a branch superposition]
\label{lem:branch_superposition}
Let $\ket{\Psi}$ be a pure state on $C\otimes Q$ of the form
\begin{equation}
\ket{\Psi}=\sum_{k=1}^{m}\sqrt{p_k}\,\ket{k}_C\ket{\psi_k}_Q,
\label{eq:branch_state}
\end{equation}
where $\{\ket{k}_C\}$ is an orthonormal set of computational-basis states and
each $\ket{\psi_k}$ is a normalized pure state on $Q$. Then the relative entropy
of coherence (with respect to the product computational basis) satisfies
\begin{equation}
\Cr(\ket{\Psi}\bra{\Psi}) = H(\{p_k\}) + \sum_{k=1}^{m} p_k\,\Cr(\ket{\psi_k}\bra{\psi_k}),
\label{eq:Cr_branch_identity}
\end{equation}
where $H(\{p_k\})=-\sum_k p_k\log_2 p_k$ is the Shannon entropy.
\end{lemma}

\begin{proof}
Since $\ket{\Psi}$ is pure, $\Cr(\ket{\Psi}\bra{\Psi})=S\!\big(\Delta(\ket{\Psi}\bra{\Psi})\big)$.
Dephasing in the product computational basis yields a diagonal distribution over
pairs $(k,j)$ with probabilities $p_k\,|\braket{j}{\psi_k}|^2$. Hence
\begin{eqnarray}
S\!\big(\Delta(\ket{\Psi}\bra{\Psi})\big)&=&H(k,j)=H(k)+H(j|k)\\
&=&H(\{p_k\})+\sum_k p_k H(\{|\braket{j}{\psi_k}|^2\})\nonumber,
\end{eqnarray}
and for each pure $\ket{\psi_k}$ one has $\Cr(\ket{\psi_k}\bra{\psi_k})=H(\{|\braket{j}{\psi_k}|^2\})$.
\end{proof}

Lemma~\ref{lem:branch_superposition} isolates a universal contribution --
\emph{branching over $m$ classical alternatives adds at most $\log_2 m$ bits of coherence}
(independent of the quantum payload), with equality for uniform branching and
incoherent $\ket{\psi_k}$.

\subsection{A general upper bound: message coherence + branching budget}
\label{subsec:general_bound}
We now show how Lemma~\ref{lem:branch_superposition} leads to a general coherence
bound for LOCC communication protocols. Suppose that within the first $i$ stages
the protocol performs measurements with outcome alphabets of sizes
$m_1,m_2,\dots,m_{L(i)}$, where $L(i)$ counts the number of measurement layers up
to stage $i$. Denote by $\Cr(\sigma_R)$ the coherence in the shared resource state
and by $\Cr(\rho_M)$ the input message coherence.

\begin{proposition}[Coherence overhead bound for LOCC-assisted protocols]
\label{prop:general_budget}
For any stage $i$ in an LOCC-assisted protocol $\Pi$ as in Eq.~\eqref{eq:rho_i_general},
\begin{equation}
\Cr(\rho_i)\;\le\; \Cr(\rho_M) + \Cr(\sigma_R) + \sum_{\ell=1}^{L(i)} \log_2 m_\ell.
\label{eq:general_budget}
\end{equation}
Consequently,
\begin{equation}
\mathcal{C}_{\max}(\Pi;\rho_M,\sigma_R)\;\le\; \Cr(\rho_M) + \Cr(\sigma_R)
+ \sum_{\ell=1}^{L(\mathrm{all})}\log_2 m_\ell,
\label{eq:general_budget_peak}
\end{equation}
where $L(\mathrm{all})$ is the total number of measurement layers in the protocol.
\end{proposition}

\begin{proof}[Proof sketch]
Insert explicit classical registers that coherently store the would-be outcomes
prior to dephasing, so that each measurement layer can be modeled as
(i) an isometry that correlates a computational-basis register $\ket{k}$ with the
quantum system and (ii) subsequent dephasing of that register.
At the coherent (pre-dephasing) level, each measurement layer produces a state of
the form Eq.~\eqref{eq:branch_state} with $m_\ell$ branches. Applying
Lemma~\ref{lem:branch_superposition} layer-by-layer shows that the additional
coherence created by the $\ell$th branching is bounded by $H(\{p_k^{(\ell)}\})
\le \log_2 m_\ell$, while the remaining term is an average of post-branch
coherences. Using convexity of $\Cr$ and additivity on tensor products yields the
stated bound.
\end{proof}

Equation~\eqref{eq:general_budget} makes precise the sense in which large
intermediate coherence can be \emph{structural} even when the implemented channel
is information-neutral (e.g., identity). The message coherence $\Cr(\rho_M)$ is
a genuine input resource, whereas the sum $\sum_\ell \log_2 m_\ell$ quantifies a
protocol-induced branching overhead associated with classical feed-forward.
Teleportation corresponds to one measurement layer with $m_1=4$, yielding the
universal $\log_2 4=2$ offset per teleportation gadget ( See
Appendix~\ref{app:two_bits}).

\subsection{Channel invariance vs. circuit cost}
\label{subsec:channel_vs_circuit}
The bound in Proposition~\ref{prop:general_budget} is deliberately orthogonal to
channel capacities. Two protocols that implement the same CPTP map $\mathcal{N}$
(e.g., identity) can have different $\mathcal{C}_{\max}$ depending on their
realization (measurement structure, classical control depth, and resource
state). This observation suggests an implementation-oriented notion of resource
cost: coherence quantifies \emph{superposition complexity} of the circuit
realizing $\mathcal{N}$, rather than any increase in transmissible information.
This perspective complements information-theoretic figures of merit (capacities,
Holevo quantities) by providing an internal dynamical accounting.
\section{Implications for quantum networks and the quantum internet}
\label{sec:qinternet}
Quantum teleportation is a foundational primitive for quantum networking \cite{Meng2025QuantumNetworks}--
entanglement distribution, entanglement swapping, and repeater-based routing
ultimately rely on Bell-type measurements and classical feed-forward.
Our coherence accounting therefore has direct implications for the internal
resource requirements of quantum-internet architectures \cite{Kumar2025QuantumInternetOverview}.

\subsection{Entanglement swapping: a coherence-cost benchmark}
\label{subsec:swapping}
Entanglement swapping consumes two Bell pairs to entangle distant nodes.
Consider two Bell pairs $\ket{\Phi^+}_{AB}\otimes \ket{\Phi^+}_{CD}$, where the
intermediate node holds qubits $B$ and $C$ and performs a Bell measurement on
$BC$. Immediately prior to this measurement, the circuit implements a coherent
branching over the four Bell outcomes. Introducing an explicit classical record
register $K$ (to store the Bell outcome coherently), the pre-dephasing state can
be written in the form of Eq.~\eqref{eq:branch_state} with $m=4$ and uniform
probabilities $p_k=1/4$ in the ideal case. By Lemma~\ref{lem:branch_superposition},
the branching contribution is
\begin{equation}
\Delta {\Cr}_{\mathrm{swap}} \le \log_2 4 = 2 \ \text{bits},
\end{equation}
with equality for uniform outcomes and incoherent conditional payloads.
Thus, entanglement swapping exhibits the same \emph{branching-entropy signature}
as teleportation: a universal two-bit coherence contribution associated with the
four-way Bell-outcome branching.

This provides a natural benchmark for network operations -- Bell-measurement-based
primitives (swapping, teleportation) require transient delocalized coherence
peaks whose magnitude is controlled by the branching alphabet size.

\subsection{Repeater chains and coherence peak scaling}
\label{subsec:repeaters}
A quantum repeater chain of $N$ elementary links distributes entanglement over a
long distance via repeated entanglement swapping at intermediate nodes.
If swapping operations at different nodes are executed sequentially, the peak
coherence at any time is dominated by a single swapping operation, and therefore
remains $O(1)$ in $N$ (up to resource-state contributions and local overheads).
By contrast, if $s$ swapping operations are executed in parallel (e.g., to
increase throughput by multiplexing), the branching budgets add. A conservative
coherence overhead estimate based on Proposition~\ref{prop:general_budget} yields
\begin{equation}
\mathcal{C}_{\max}^{(\mathrm{parallel\ swaps})}
\;\lesssim\; \Cr(\sigma_R) + 2s,
\label{eq:parallel_swaps}
\end{equation}
where $\sigma_R$ denotes the joint state of the entanglement resources used in
that time window. This estimate highlights a design tradeoff: multiplexing can
increase entanglement distribution rates but requires hardware capable of
supporting larger transient coherence peaks (and thus greater vulnerability to
dephasing during those stages).

\subsection{Noisy links: coherence as an internal fragility diagnostic}
\label{subsec:noisy_network}
In realistic networks, distributed Bell pairs are imperfect and local operations
and memories are noisy. While standard performance metrics focus on output
quantities (fidelity, entanglement rate, QBER\footnote{The fraction of bits that are received incorrectly in a quantum communication protocol.}), stage-resolved coherence provides
an internal diagnostic that localizes \emph{where} the protocol is most fragile.
In particular, dephasing noise suppresses delocalized coherence at intermediate
stages without necessarily changing the intended channel action at the ideal level.
Within the present framework, this effect appears as a reduction of
$\mathcal{C}_{\max}$ and as a flattening of the coherence profile near the
pre-measurement stages where branching occurs.

More concretely, if a swapping (or teleportation) stage is followed by a
dephasing channel $\mathcal{D}_\lambda$ acting on the relevant registers, then
by data processing of quantum relative entropy one has
\begin{equation}
\Cr\!\left(\mathcal{D}_\lambda(\rho)\right)
= D\!\left(\mathcal{D}_\lambda(\rho)\,\|\,\Delta\mathcal{D}_\lambda(\rho)\right)
\le D\!\left(\rho\,\|\,\Delta(\rho)\right)=\Cr(\rho),
\label{eq:dephasing_monotone}
\end{equation}
so coherence decreases monotonically under physical dephasing noise. Tracking
this decrease stage-by-stage yields a principled way to identify the coherence
bottlenecks that dominate network fragility, complementing end-to-end fidelity
metrics.

\subsection{Operational relevance for a quantum internet}
\label{subsec:operational_relevance}
In a more speculative direction, the above considerations suggest that coherence accounting can inform quantum
internet \cite{Kumar2025QuantumInternetOverview}  engineering in at least two ways. First, it provides an
implementation-oriented resource metric that complements channel capacities --
protocols that realize the same end-to-end channel (e.g., teleportation-based
routing approximating identity transmission) can differ substantially in their
peak coherence costs depending on scheduling, multiplexing, and classical control
structure. Second, coherence profiles provide a diagnostic lens for noise --
by identifying where transient delocalized coherence is concentrated, one
pinpoints the stages where error mitigation, improved memory coherence times, or
hardware-level dynamical decoupling will have the largest impact. Of course, all of those considerations must be carefully investigated to be confirmed.

\section{Conclusion}
In this work we have developed a coherence-based perspective on quantum
communication protocols, focusing on how coherence is generated, redistributed,
and consumed during protocol execution. Using the relative entropy of coherence
as a circuit-level diagnostic, we have shown that coherence captures dynamical
aspects of quantum communication that are not visible through entanglement-based
analyses alone.

Our comparative study of superdense coding and quantum teleportation reveals a
striking contrast. Multipartite resource states capable of supporting scalable
superdense coding can exhibit only logarithmic or even constant coherence
scaling, highlighting that coherence need not track communication capacity in a
straightforward way. In sharp contrast, teleportation displays an unavoidable
and extensive coherence cost that grows linearly with the number of teleported
qubits. Through a stage-resolved analysis, we have shown that this cost is
intrinsic to the protocol itself and reflects the global superposition structure
required to implement an identity channel on an unknown quantum state.

Beyond these specific examples, we have formulated a general coherence
accounting framework for LOCC-assisted quantum communication protocols. By
explicitly isolating message-dependent contributions from protocol-induced
branching costs, we derived general bounds showing that large intermediate
coherence can arise solely from the structure of classical feed-forward and
measurement-induced branching, even when the implemented channel is
information-neutral. This establishes protocol-induced coherence as a structural
property of circuit realizations rather than a direct indicator of transmitted
information.

Importantly, we have demonstrated explicitly that extensive intermediate
coherence generation is fully consistent with information-theoretic bounds,
including the Holevo limit. The coherence quantified in this work characterizes
transient global superpositions that are operationally inaccessible due to
measurement-induced dephasing and does not correspond to an increase in
accessible classical information. This clarifies the conceptual distinction
between coherence as a circuit-level resource and information capacity as
constrained by communication theory.

Our results further suggest that coherence provides a useful lens for analyzing
quantum network primitives relevant to a quantum internet, including
teleportation-based routing, entanglement swapping, and repeater architectures.
In such settings, coherence accounting offers an internal diagnostic of
superposition complexity and noise sensitivity that complements standard
end-to-end metrics such as fidelity and rate.

More broadly, coherence as quantified here does not merely characterize static
properties of quantum states, but acts as a measure of the superposition
complexity required to implement a given quantum communication primitive at the
circuit level. Extensions of this framework to noisy channels, imperfect
resource states, and fault-tolerant architectures may therefore provide valuable insight into coherence overheads and
circuit-level dynamical constraints of realistic quantum technologies.

\appendix
\section{Two-bit coherence offset}
\label{app:two_bits}

Here we provide an explicit derivation of the universal two-bit coherence offset that appears in the teleportation scaling
law, Eq.~\eqref{eq:scaling}. We consider a single teleportation gadget acting on an incoherent input state $\ket{0}$ and a
shared Bell pair $\ket{\Phi^+}=(\ket{00}+\ket{11})/\sqrt{2}$. At the stage immediately prior to measurement, after the
application of the CNOT and Hadamard gates on Alice’s side, the global three-qubit state (message qubit plus Bell pair)
takes the form
\begin{equation}
\ket{\Psi}=\frac{1}{2}\big(\ket{000}+\ket{011}+\ket{100}+\ket{111}\big),
\label{eq:tp_state}
\end{equation}
where the ordering of subsystems is understood as message qubit, Alice’s ancilla, and Bob’s qubit.

The state in Eq.~\eqref{eq:tp_state} has four computational-basis components with equal amplitudes. Its density matrix
therefore satisfies $S(\rho)=0$, while the fully dephased state $\rho_{\diag}$ has four nonzero diagonal elements equal to
$1/4$, yielding
\begin{equation}
S(\rho_{\diag}) = -\sum_{k=1}^{4} \frac{1}{4}\log_2\frac{1}{4} = 2.
\end{equation}
Using the definition of the relative entropy of coherence, $\Cr(\rho) = S(\rho_{\diag}) - S(\rho)$, we obtain $\Cr=2$ bits.
This value is independent of the input message as long as the input is incoherent in the computational basis, and it
depends only on the structure of the teleportation circuit at the maximal-coherence stage. As a result, each teleportation
gadget contributes a fixed two-bit coherence offset, leading to the linear scaling identified in Eq.~\eqref{eq:scaling} for
parallel teleportation of $n$ qubits.

\begin{acknowledgments}
M.C.O. acknowledges partial financial support from the National Institute of Science and Technology for Applied Quantum
Computing through CNPq (Process No.~408884/2024-0) and from the São Paulo Research Foundation (FAPESP), through the Center
for Research and Innovation on Smart and Quantum Materials (CRISQuaM, Process No.~2013/07276-1).
\end{acknowledgments}

\bibliography{references}

\end{document}